%% file: bciot-be-ctde-arxiv.tex
\pdfoutput=1
\documentclass[journal]{IEEEtran}

\usepackage{cite}
\usepackage{graphicx}
\graphicspath{{fig/}}
\usepackage{amsmath}
\usepackage{algorithmic}
\usepackage{array}
\usepackage[caption=false,font=footnotesize]{subfig}
\usepackage{url}

\usepackage{amssymb}
\usepackage{amsthm}
\usepackage{multicol}
\usepackage{multirow}
\usepackage{mdwmath}
\usepackage{mdwtab}

\newtheorem{theorem}{Theorem}
\newtheorem{lemma}{Lemma}
\newtheorem{corollary}{Corollary}
\newtheorem{proposition}{Proposition}
\newtheorem{auss}{Assumption}

\usepackage{booktabs}
\usepackage{makecell}

\begin{document}

\title{BE-CTDE: Blockchain-Enabled Multi-Agent Reinforcement Learning with Anomaly Detection and Token Economics}

\author{Chih-Wei Huang, Hsin-Chieh Lee, Yun-Chen Yu, Ibrahim Althamary, Rongxing Lu, and Phone Lin%
\thanks{This work has been submitted to the IEEE for possible publication. Copyright may be transferred without notice, after which this version may no longer be accessible.}%
\thanks{(Corresponding author: Chih-Wei Huang.)}%
\thanks{Chih-Wei Huang, Hsin-Chieh Lee, and Yun-Chen Yu are with the Department of Communication Engineering, National Central University, Taoyuan City 320317, Taiwan (e-mail: cwhuang@ce.ncu.edu.tw; 112523042@cc.ncu.edu.tw; 113523038@cc.ncu.edu.tw).}%
\thanks{Ibrahim Althamary is with the Interdisciplinary Research Center for Intelligent Secure Systems, King Fahd University of Petroleum and Minerals (KFUPM), Dhahran 31261, Saudi Arabia (e-mail: ibrahim.thamary@kfupm.edu.sa).}%
\thanks{Rongxing Lu is with the School of Computing, Queen's University, Kingston, ON K7L 2N8, Canada (e-mail: rongxing.lu@queensu.ca).}%
\thanks{Phone Lin is with the Department of Computer Science and Information Engineering, the Graduate Institute of Networking and Multimedia, the Graduate Institute of Biomedical Electronics and Bioinformatics, the Internet of Things Research Center, and the Graduate Institute of Medical Device and Imaging, National Taiwan University, Taipei City 10617, Taiwan (e-mail: plin@csie.ntu.edu.tw).}%
}


\maketitle

\begin{abstract}
Multi-agent reinforcement learning supports cooperative decision-making in Internet of Things systems, but strategic deviations can disrupt coordination and undermine confidence in the training process. This paper proposes Blockchain-Enabled Centralized Training and Decentralized Execution (BE-CTDE), which uses blockchain as an active training-management layer. A behavioral verifier compares executed trajectories with shared-policy recommendations; its flags drive a token mechanism combining slashing, burning, rewards, and compensation. Retained balances determine each agent's contribution to the global policy update, connecting auditable behavior to learning influence. We derive a one-round incentive condition that accounts for action-dependent compensation and the payoff from deviation, and a finite expected exclusion-time bound under explicit conditional detection and accounting assumptions, showing perfect detection is unnecessary. An exact finite-population capacity bound characterizes the limits of the population-relative detector. In multi-agent path finding experiments, BE-CTDE outperforms a state-of-the-art learning-based path planner with substantially higher success rates under attack, even when adversaries form the majority. It also outperforms partially defended baselines in all configurations, with collision rates up to $35\%$ lower than under detection without economic consequences and up to $43\%$ lower than under conventional stake slashing. Together, these results characterize how behavioral verification and token feedback can support accountable cooperative learning.
\end{abstract}


\section{Introduction}

The proliferation of the Internet of Things (IoT) drives applications such as smart logistics, urban traffic management, and autonomous drone swarms, in which many intelligent agents collaborate in dynamic environments to accomplish complex tasks~\cite{zhou2021smarts}. These interconnected agents share communication, storage, and computing resources for coordinated operations such as adaptive traffic-signal control, infrastructure monitoring, and synchronized drone missions~\cite{chen2024vehicle}. Achieving high performance, security, and reliability in such resource-constrained, interactive settings remains challenging, particularly amid data tampering, privacy breaches, and adversarial disruption~\cite{aouedi2024survey}.

To address decision-making complexity in these settings, Multi-Agent Reinforcement Learning (MARL) has emerged as an essential paradigm in which agents optimize resource scheduling and actions through continuous learning over shared information~\cite{he2023toward, gao2022survey}. A widely adopted framework is Centralized Training and Decentralized Execution (CTDE), which derives global insights from centralized training while preserving the distributed execution that real-world IoT latency and connectivity constraints demand. Learning-based multi-agent path finding (MAPF) commonly adopts this framework~\cite{alkazzi2024comprehensive}, as does the state-of-the-art SCRIMP planner~\cite{wang2023scrimp}.

The traditional training process of the CTDE global model introduces critical vulnerabilities. These include untraceable data that hinder transparency, inherent weaknesses in the centralized architecture such as lack of accountability and a single point of failure~\cite{gao2022survey}, and policy disobedience where agents may execute malicious actions~\cite{arnold2022only}. In open IoT deployments, these weaknesses are exacerbated by the broader threat landscape facing autonomous AI agents, including data poisoning, model evasion, and behavioral manipulation~\cite{deng2025ai}. Purely on-device training, conversely, loses the global context cooperative tasks require. What is needed is a training framework that distributes management of the learning process for transparency and fault tolerance without sacrificing the coordination benefits of centralized training.

Blockchain can provide a replicated management layer~\cite{khan2022internet, puthal2018everything}, in which commitments to model updates and interactions support an auditable coordination history~\cite{meijers2022blockchain, cao2022decentralized}. Commitments bind recorded inputs; they do not establish sensor truth or the correctness of off-chain inference. Such deployments commonly span several operators with no single party that all of them trust, which is the setting in which a replicated ledger repays its overhead. Trust is then bounded rather than eliminated, since the shared coordination infrastructure is trusted while the agents that participate in training are not. Blockchain alone, however, does not address policy disobedience, because it preserves submitted interaction records yet cannot determine whether recorded behavior reflects genuine cooperation or adversarial interference. Closing this loop requires behavioral verification and incentive alignment inside the training process itself, complementing recent advances in adversarial detection for cooperative multi-agent systems~\cite{li2023among, zhao2024made}.

Several architectures have been proposed at the convergence of blockchain and AI. Early frameworks pioneered decentralized machine learning on simulated blockchains, using token incentives to validate data quality~\cite{harris2019decentralized, harris2020analysis, weng2019deepchain}. Subsequent work added incentives for data and compute sharing on real chains but largely overlooked behavior verification and adversarial mitigation during training~\cite{zanardo2022learningchain, zhao2025blockchain}. In blockchain-integrated MARL specifically, recent efforts explore on-chain coordination for multi-object tracking, path planning in cognitive vehicular networks, and resource management for blockchain-based federated learning in IoT~\cite{shen2023blockchain, chang2023distributed, mishra2024novel}. The most directly comparable framework, LearnChain~\cite{sami2024learnchain}, makes MARL model updates traceable, but its record is compiled after the fact, and an agent logged as deviant still contributes to the shared model exactly as a compliant one does. Across this literature the ledger observes training rather than participating in it, so the step from detecting misbehavior to withholding its influence over the learned policy remains open.

The gap is therefore not auditability but consequence, because evidence of how an agent actually behaved is not carried back into the update that trains the shared policy. We close that loop with Blockchain-Enabled Centralized Training and Decentralized Execution (BE-CTDE), in which blockchain is an active training-management layer. The contributions are:
\begin{itemize}
    \item \textbf{Closed-loop BE-CTDE.} On-chain coordination and off-chain learning form an auditable training cycle: behavioral verification updates economic state, economic state determines aggregation weight, and insufficient balance ends an agent's training contribution.
    \item \textbf{Detection and aggregation.} A framework-agnostic verifier compares executed trajectories with shared-policy recommendations using a task-specific residual. The coupled token mechanism applies slashing, burning, rewards, and compensation, while retained balances weight the per-agent learning contributions. This makes detected deviations economically consequential without requiring gradient-space similarity between heterogeneous agents.
    \item \textbf{Conditional theoretical guarantees.} We derive the one-round incentive comparison with action-dependent compensation and deviation utility. Exclusion in finite expected time needs no perfect detector, only history-uniform detection above a threshold, at most $1/3$ in our setting even with full compensation, and the excluded adversary then contributes no new gradient. An exact capacity result makes the threshold hold for every retained population.
    \item \textbf{Implementation and MAPF evaluation.} A private Ethereum/Solidity implementation evaluates BE-CTDE across population scales, adversary ratios, and obstacle densities. BE-CTDE achieves higher success rates under attack than the state-of-the-art SCRIMP planner and, in every configuration, lower collision rates than partially defended baselines, by up to $35\%$ relative to detection without economic consequences and up to $43\%$ relative to conventional stake slashing.
\end{itemize}

The rest of the paper is organized as follows. Section~\ref{sec:related-work} reviews related work. Section~\ref{sec:system-model} sets out the system model, threat model, and design problem, and Section~\ref{sec:detection-incentive} develops the detection, incentive, and aggregation mechanisms with their guarantees. Section~\ref{sec:implementation} describes the blockchain implementation, Section~\ref{sec:performance-evaluation} reports the MAPF evaluation, and Section~\ref{sec:discussion-conclusion} concludes the paper.

\section{Related Work}\label{sec:related-work}

\subsection{Blockchain Foundations for IoT}
\label{sec:blockchain-foundations}

Blockchain replicates an append-only, cryptographically linked ledger across nodes, replacing trust in a central operator with consensus-backed immutability~\cite{puthal2018everything}, so transactions and model-update commitments can be audited by every participant when multiple IoT operators share infrastructure without a common trusted server. Proof of Work relies on energy-intensive computational competition, whereas Proof of Stake (PoS) uses stake-based validator selection; their energy requirements differ~\cite{wendl2023environmental}. Ethereum also supports deterministic smart contracts executed by the Ethereum Virtual Machine (EVM)~\cite{wohrer2018smart}, encoding deposits, verifiable artifacts, rewards, and penalties as an automated trust layer. BE-CTDE uses a permissioned Ethereum management layer to coordinate training while neural-network execution remains off-chain.

\subsection{AI and Blockchain Integration}
\label{sec:ai-blockchain-integration}

Research on decentralized AI increasingly uses blockchain as a coordination substrate. Table~\ref{tab:functional-comparison} compares representative systems along the dimensions that separate them from BE-CTDE.

\begin{table*}
\centering
\scriptsize
\setlength{\tabcolsep}{3pt}
\renewcommand{\arraystretch}{1.15}
\caption{Functional comparison of AI and blockchain integration across recent literature}
\begin{tabular}{c c c c c c c c c}
\toprule
\textbf{References} &
\textbf{Blockchain} &  
\textbf{MARL} & 
\makecell[c]{\textbf{On-chain}\\\textbf{Parameters}} & 
\makecell[c]{\textbf{On-chain}\\\textbf{Learning}} & 
\makecell[c]{\textbf{Incentive}\\\textbf{Mechanism}} & 
\makecell[c]{\textbf{Token-based}\\\textbf{Economy}} & 
\makecell[c]{\textbf{Data}\\\textbf{Verification}} & 
\makecell[c]{\textbf{Malicious} \\\textbf{Agent}\\\textbf{Detection}} \\
\midrule
\cite{harris2019decentralized, harris2020analysis} 
& Simulator & – & \checkmark & \checkmark & \checkmark & \checkmark & \checkmark & – \\
\hline
\cite{zhao2025blockchain} 
& Simulator & – &  \checkmark& – & \checkmark & \checkmark & \checkmark & \checkmark \\
\hline
\cite{weng2019deepchain} 
& Real & – &  \checkmark& – & \checkmark & \checkmark & \checkmark & \checkmark \\
\hline
\cite{shen2023blockchain} 
& Simulator  & \checkmark &  \checkmark &  –& \checkmark & –  & \checkmark & – \\
\hline
\cite{zanardo2022learningchain} 
& Real   & – & \checkmark & \checkmark & \checkmark & \checkmark & – & – \\
\hline
\cite{chang2023distributed} 
& Simulator  & \checkmark &\checkmark  &–  & \checkmark & – & \checkmark & \checkmark \\
\hline
\cite{mishra2024novel} 
& Simulator  & – & \checkmark & – & \checkmark & \checkmark & \checkmark & – \\
\hline
\cite{sami2024learnchain} 
& Real  & \checkmark & \checkmark & \checkmark & \checkmark & – & \checkmark & – \\
\hline
\textbf{BE-CTDE}
& Ethereum & \checkmark & Hash & Audit & \checkmark & Ledger & Inputs & Flags \\
\bottomrule
\multicolumn{9}{@{}l}{\emph{On-chain learning} denotes anchored artifacts, not neural-network execution; BE-CTDE anchors accounting and input/policy commitments.}
\end{tabular}
\label{tab:functional-comparison}
\end{table*}

Early systems introduced token-incentivized data validation on Ethereum simulators~\cite{harris2019decentralized,harris2020analysis} and later extended it to auditable deep learning and resource sharing on real chains~\cite{weng2019deepchain,zanardo2022learningchain}. IoT-focused work primarily addresses federated learning, with privacy-preserving data sharing~\cite{liu2023pd2s}, blockchain resource management~\cite{mishra2024novel}, privacy plus incentives for the Internet of Vehicles~\cite{zhao2025blockchain}, and reputation-weighted token incentives for IoT ecosystems~\cite{kuang2026brite}. These systems address realistic deployment constraints but not MARL's cooperative dynamics, behavioral verifiability, or trajectory-level adversarial actions.

A smaller literature directly couples blockchain with MARL. Shen et al.~\cite{shen2023blockchain} coordinate collaborative multi-object tracking, and Chang et al.~\cite{chang2023distributed} address path planning and scheduling in cognitive vehicular networks. Both demonstrate feasibility but use offline or simulated chains and emphasize reward traceability. More recent work extends the pairing to computational task offloading in IoT~\cite{chen2026swarm} and to decentralized multi-agent path planning with federated updates~\cite{pasupathy2026decentralized}. LearnChain~\cite{sami2024learnchain}, the closest comparison, records MARL updates on a permissioned chain, yet the chain serves as an audit endpoint, so an agent's recorded conduct never alters its influence over the aggregated model.

Tian~\cite{tian2025blockchain} more directly records agent behavior and automates penalties and rewards in a smart contract. BE-CTDE differs in both its evidence and its feedback: a trajectory residual compares policy-recommended and executed actions, and the resulting token balance enters the gradient step as an aggregation weight. Thus, detection changes economic state, economic state changes influence, and sustained balance decay triggers exclusion. LearnChain records updates without inferring behavior, while Zhao et al.~\cite{zhao2025blockchain} slash stake without feeding the penalty into the update rule. As Table~\ref{tab:functional-comparison} summarizes, BE-CTDE combines executable EVM accounting, residual-based detection (Section~\ref{sec:detection}), token accounting (Section~\ref{sec:incentive}), and token-weighted aggregation (Section~\ref{sec:aggregation}).
\subsection{Byzantine-Resilient Multi-Agent Learning and Incentive Design}
\label{sec:byzantine-related}

Robust gradient aggregation offers a different defense. Krum~\cite{blanchard2017krum} replaces averaging with a selection rule based on the $n-f-2$ nearest gradients and provides convergence under Byzantine workers; coordinate-wise median and trimmed-mean rules attain optimal statistical rates~\cite{yin2018byzantine}. Both families assume an honest majority, so their guarantees do not extend to the majority-adversary setting evaluated in Section~\ref{sec:performance-evaluation}. These methods identify adversaries through statistical deviation in parameter space, assuming honest gradients concentrate around a common direction. In MARL, however, agents face heterogeneous states, observations, and goals, so honest gradients need not concentrate. BE-CTDE instead evaluates each executed trajectory against the shared policy recommendation, a behavioral reference meaningful across distinct local sub-problems.

Work specific to cooperative MARL includes Li et al.'s Bayesian Adversarial Robust Dec-POMDP~\cite{li2024byzantine}, in which Byzantine agents are latent types and honest agents maintain beliefs about peers. This approach is theoretically rigorous but requires extensive belief tracking and supplies no economic mechanism for sustained honest behavior. BE-CTDE takes a complementary path: trajectory-level verification changes accounting state and gradient weights, with removal conditional on crossing an exclusion threshold.

BARD-MARL combines policy-graph evidence with Bayesian trust signals to diagnose compromised learned communication~\cite{murimi2026bard}. Wu and Klabjan study blockchain-based federated learning with loss-based softmax aggregation and a convergence result in a convex setting~\cite{wu2023robust}. These works distinguish detection and robust update selection from the closed feedback path studied here: movement evidence changes token accounts, accounts change gradient weights, and insufficient balance ends training contribution. Our experiments do not benchmark their algorithms or transfer their guarantees to MARL.

BE-CTDE also draws on mechanism design~\cite{myerson1981optimal} and incentive-compatible federated learning~\cite{tu2022incentive,xuan2024contract}. Section~\ref{sec:incentive} makes explicit the compensation comparison and external deviation payoff required for a one-round incentive result. A separate conditional drift argument concerns persistent disruption even when token rewards do not deter it.

\section{System Model and Problem Formulation} 
\label{sec:system-model}

\subsection{System Architecture}
\label{sec:architecture}

Figure~\ref{fig:architecture} presents BE-CTDE's two-layer architecture: IoT agents execute the learned policy, while a blockchain trust layer coordinates training and verifies behavior.

\begin{figure}
 \centering \includegraphics[width=\columnwidth]{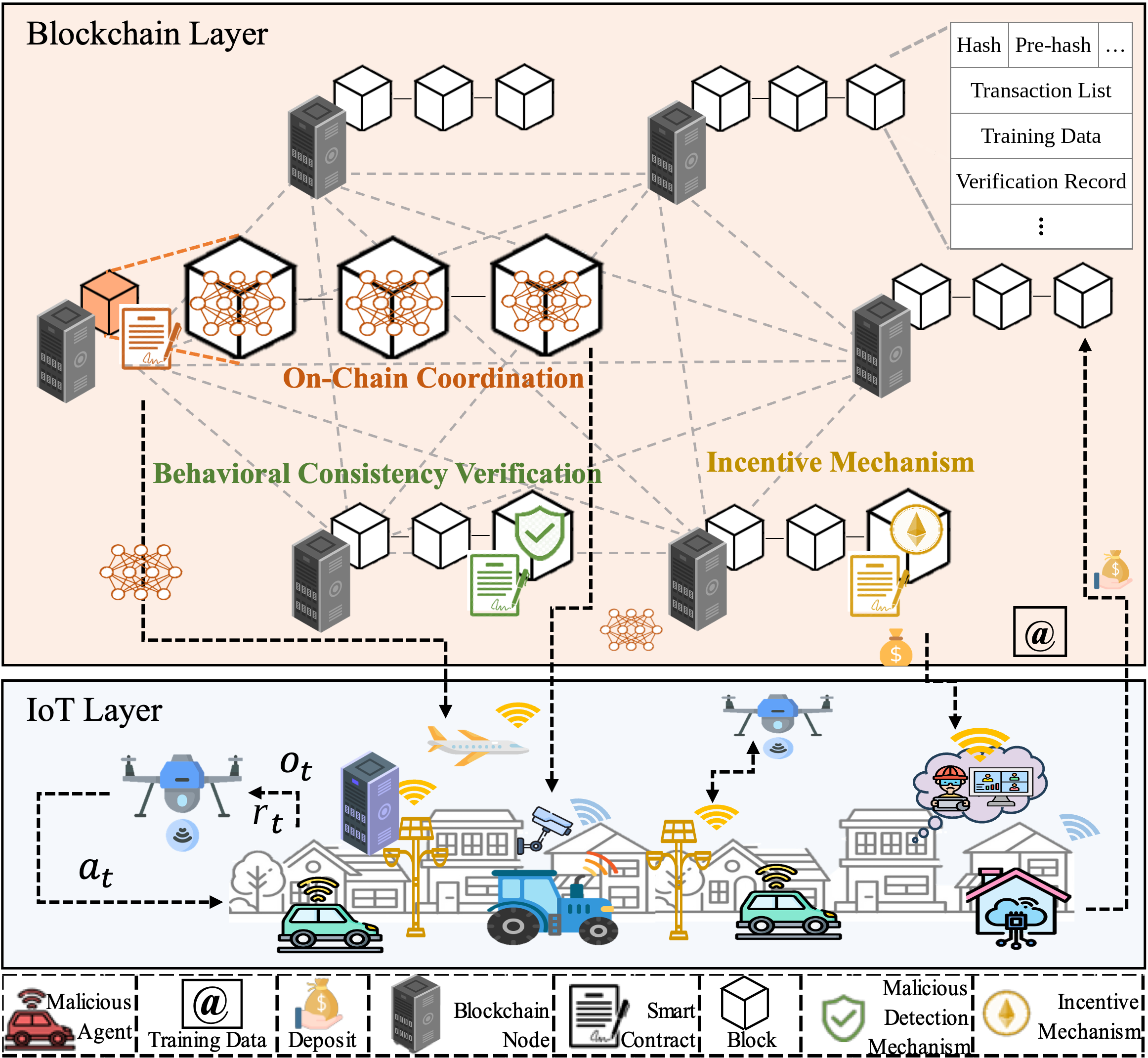}
 \caption{BE-CTDE system architecture: IoT agents interact with the environment through observations $o_t$, rewards $r_t$ and actions $a_t$, while the blockchain layer coordinates training, verifies behavior and manages token incentives. Neural-network computation remains off-chain.}
 \label{fig:architecture}
\end{figure}

The IoT layer supplies observations and executes a shared policy. During training, a trusted coordinator records physical transitions, computes the policy recommendation, and submits sampled records to the verification contract. The contract evaluates consistency and updates token accounts; retained balances determine off-chain gradient weights. This is an auditable training interface, not a proof that independently reported sensor data are truthful. Table~\ref{tab:notation} summarizes the notation.

\begin{table}[t]
\caption{Principal Notation}
\label{tab:notation}
\centering
\footnotesize
\renewcommand{\arraystretch}{1.05}
\setlength{\tabcolsep}{2pt}
\begin{tabular}{@{}>{\raggedright\arraybackslash}p{0.27\columnwidth}
                    >{\raggedright\arraybackslash}p{0.68\columnwidth}@{}}
\toprule
\textbf{Symbol} & \textbf{Definition} \\
\midrule
\multicolumn{2}{@{}l}{\textit{System and learning}} \\
$N,n_t,M,M_t$ & Initial, active, initially malicious, and active malicious counts ($M<N$) \\
$\mathcal K_t,\mathcal H,\mathcal M,f_i(t)$ & Retained, honest, malicious sets and detection flag \\
$\mathcal{S},\mathcal{A},\mathcal{O}$ & State, joint-action, and joint-observation spaces \\
$\Omega,\mathcal{P},\mathcal{R},\gamma$ & Observation map, transition/reward functions, and discount factor \\
$(o,a,r,m)_{i,t}$ & Observation, action, reward, and prior message of agent $i$ \\
$\pi_\theta,\theta,J(\theta)$ & Global policy, parameters, and executed return \\
$g_i(t),\alpha,T,E$ & Local gradient, learning rate, horizon, and episode count \\
\midrule
\multicolumn{2}{@{}l}{\textit{Detection and token economy}} \\
$a_{i,t}^{*},a_{i,t}^{\mathrm{exec}}$ & Policy-recommended and executed actions \\
$\phi_i,\rho$ & Outcome map and outcome metric \\
$d_{i,t},D_i$ & Per-step and cumulative behavioral residual \\
$h,L,\mathcal T$ & Sampling interval, episode length, and sampling instants \\
$\mu_t,\sigma_t,k,\tau_t$ & Active-population statistics, sensitivity, threshold \\
$F_{\max},\kappa_N,C_{\max}$ & Flag capacity, worst compensation ratio and bound \\
$P_{\mathrm{tp}},P_{\mathrm{fp}}$ & Detector true- and false-positive rates \\
$s,\beta$ & Deposit and reward factor \\
$\lambda,c_t$ & Compensation factor and per-recipient compensation \\
$S_i(t),S_{\min},\omega_i(t)$ & Token balance, exclusion threshold, and aggregation weight \\
$p_H,p_D,b_H,b_D$ & Conditional flag probabilities and unflagged gains, honest ($H$) versus deviation ($D$) \\
$\underline p,b_{\max},P^{\ast},\varepsilon$ & Detection floor, gain bound, detection threshold, and drift margin \\
$T_j^{\mathrm{ex}},\Delta,u_{\mathrm{dev}}$ & Exclusion time, incentive margin, and per-round deviation payoff \\
\bottomrule
\end{tabular}
\end{table}

\subsection{Cooperative Decision-Making under Partial Observability}
\label{sec:dec-pomdp}

The cooperative decision-making process across $N$ IoT agents is modeled as a Decentralized Partially Observable Markov Decision Process (Dec-POMDP) $\langle \mathcal{S}, \mathcal{A}, \mathcal{O}, \Omega, \mathcal{P}, \mathcal{R}, \gamma \rangle$. The state space $\mathcal{S}$ captures the global configuration of agents and environmental factors such as obstacles or resource loads; the joint action space $\mathcal{A} = \mathcal{A}_1 \times \cdots \times \mathcal{A}_N$ comprises the discrete actions $a_{i,t} \in \mathcal{A}_i$ of each agent, such as movement directions or resource bids. Since agents cannot perceive the global state directly, the observation function $\Omega: \mathcal{S} \rightarrow \mathcal{O}$ yields a local observation $o_{i,t} \in \mathcal{O}_i$ correlated with each agent's sensing range and communication history, with $\mathcal{O} = \mathcal{O}_1 \times \cdots \times \mathcal{O}_N$ the joint observation space. Dynamics follow the transition function $\mathcal{P}(s' \mid s, a)$, and agents jointly maximize the discounted return $\mathbb{E}\!\left[\sum_{t=0}^{T} \gamma^t r_t\right]$ under $\mathcal{R}: \mathcal{S} \times \mathcal{A} \rightarrow \mathbb{R}$, where $r_t$ aggregates task-specific extrinsic terms and intrinsic exploration bonuses and $\gamma \in [0, 1)$ is the discount factor.

Each agent learns a parameterized policy $\pi_\theta(a \mid o)$ executed in a decentralized manner. Following the CTDE paradigm, the policy parameters $\theta$ are updated centrally using globally aggregated information; in BE-CTDE, a trusted off-chain learner consumes the token-accounting outcomes. Recorded hashes identify the policy and inputs used; they do not prove that the learner performed correct optimization.

\subsection{Threat Model}
\label{sec:threat-model}

The intended deployment separates participating agents from a permissioned training coordinator and execution ledger. Up to $M<N$ agents can deviate from policy actions~\cite{arnold2022only, li2024byzantine}; the remaining agents sample or execute the prescribed policy. The model assumes authentic physical observations and an honest training coordinator. The permissioned ledger provides an auditable coordination history under its consensus assumptions~\cite{puthal2018everything}.

\emph{Adversarial capabilities and objective.} Attackers can replace their sampled actions with persistent blocking, occupation of another agent's goal, or repetitive motion~\cite{ivanova2021adversarial}. Their objective is disruption rather than token profit. Attacker identities are used to define the adversarial conditions and evaluate outcomes; the behavioral decision rule operates on trajectory evidence. Exclusion removes an agent's contribution to training; it does not physically remove the agent from the environment. Consequently, learning-layer containment does not guarantee that honest goals remain physically attainable.

\emph{Trust and scope.} The global policy is computed by an honest off-chain learner. A policy hash commits to its bytes, not to correct inference, and a transaction signature authenticates the submitting controller, not independent agent observations. The experiments assess learning under behavioral attacks; network consensus faults and validator incentives are outside this evaluation. Sybil attacks, fabricated sensor records, gradient-inversion attacks, collusion channels, and concealed nonstationary attacks are outside the evaluated threat model; see Section~\ref{sec:discussion-conclusion}.

\subsection{Design Problem}
\label{sec:design-problem}

The policy samples joint actions, but the environment advances under the executed joint action $a_{t}^{\mathrm{exec}}$, which the adversaries of Section~\ref{sec:threat-model} partly control. The objective is therefore the executed return
\begin{equation}
J(\theta) = \mathbb{E}\!\left[\sum_{t=0}^{T} \gamma^{t}\, \mathcal{R}\!\left(s_t,\, a_{t}^{\mathrm{exec}}\right)\right].
\label{eq:executed-return}
\end{equation}
Adversarial behavior can depress $J(\theta)$ through two channels: it perturbs the executed trajectory, and the resulting training data can distort the centralized update. BE-CTDE directly changes contribution to the second channel. It can improve learned responses to interference, but cannot physically remove an obstacle or guarantee access to an occupied goal.

The design couples behavioral consistency, token accounting and policy aggregation. Detection supplies a task-level reference for identifying deviations; the economic rule changes the cost of those deviations; and retained balances determine learning influence. The analysis asks when a deviation is deterred and when negative balance drift is sufficient for exclusion. The MAPF evaluation measures the resulting task performance across agent populations, adversary ratios and obstacle densities. Task metrics alone do not establish the conditional detection assumptions of the analysis.

\section{On-Chain Detection and Token-Based Incentive Mechanisms}
\label{sec:detection-incentive}

The blockchain layer governs training through sampled-input commitments, behavioral verification, and token incentives while high-frequency execution remains in the IoT layer.

\subsection{Training Workflow Overview}
\label{sec:training-workflow}

Figure~\ref{fig:workflow} summarizes the coupled on-chain/off-chain training workflow.

\begin{figure}
 \centering \includegraphics[width=\columnwidth]{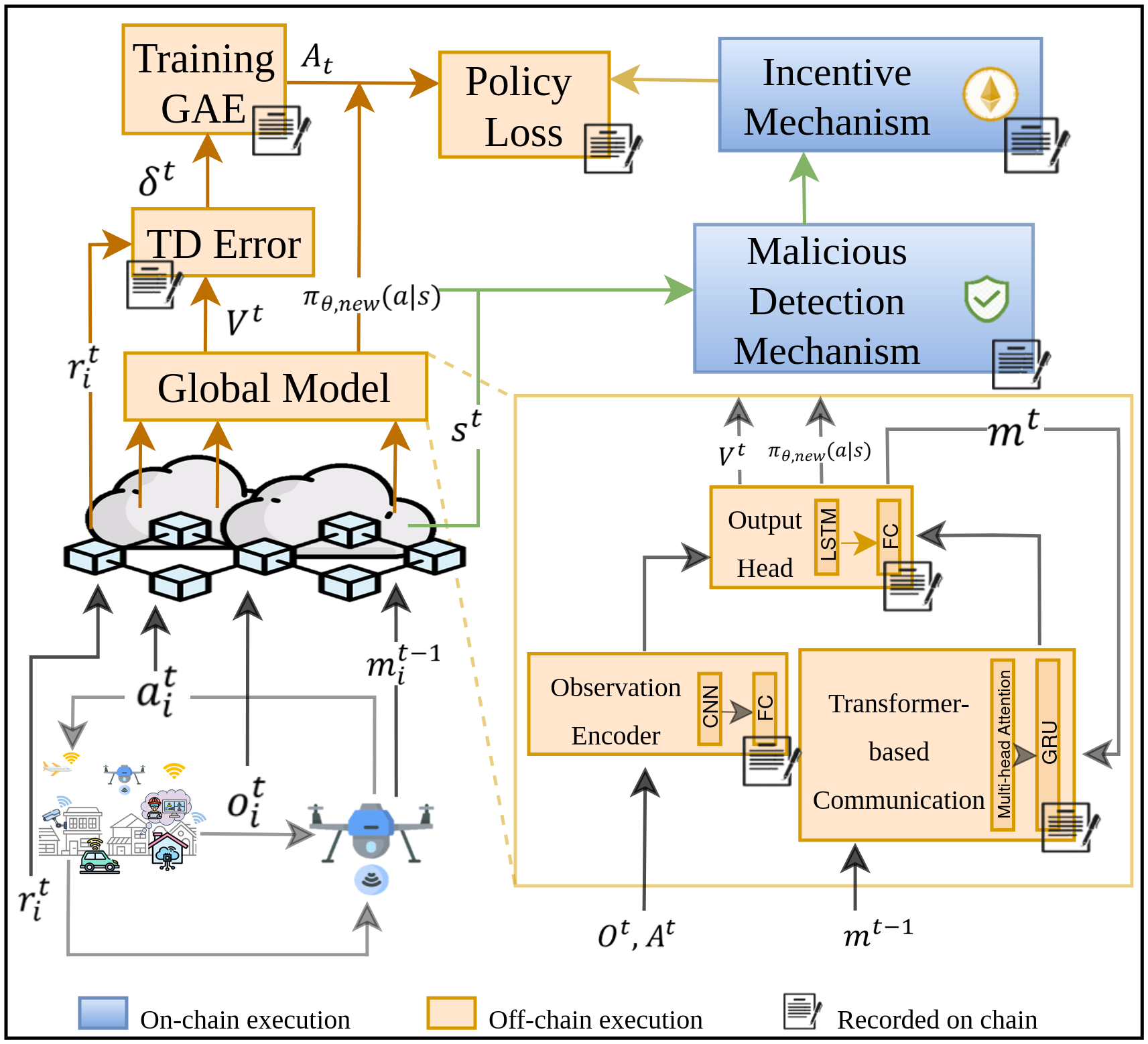}
 \caption{BE-CTDE training workflow and shared-policy architecture. Neural inference, communication, advantage estimation and PPO run off-chain (orange), while behavioral verification and token settlement run on-chain (blue); document icons mark results committed to the chain.}
 \label{fig:workflow}
\end{figure}

Each agent encodes its local observation $o_i^t$ together with the peer messages $m_i^{t-1}$, emits $m^t$, and reads a value estimate $V^t$ and an action distribution $\pi_{\theta,\mathrm{new}}(a\mid s)$ from the shared model. Every $h$ steps, the sampled state $s^t$ is recorded as the interaction evidence for behavioral verification, so an episode of length $L$ supplies at most $\lfloor L/h\rfloor$ samples. At every step, by contrast, the off-chain SCRIMP/PPO learner~\cite{wang2023scrimp} turns the executed actions $a_i^t$ and returned rewards $r_i^t$ into the TD error $\delta^t$ and the advantage $A_t$ that form the policy loss. Episode-end detection and incentive settlement then supply the weights applied to that loss in the next update, separating high-frequency inference from economic coordination. The analysis below concerns settlements with residual evidence; without evidence, its balance recursion makes no change.

\subsection{Framework-Agnostic Behavioral Consistency Verification}
\label{sec:detection}

Let $\mathcal K_t$ be the retained training participants before settlement $t$, with $n_t=|\mathcal K_t|\le N$. The shared policy supplies a reference action $a^*_{i,u}$, fixed by $\pi_\theta$ and the recorded observation $o_{i,u}$ alone. Verification compares what that recommendation would have produced with what the agent produced, through a task-specific outcome map $\phi_i$ carrying a state-action pair to an observable outcome, and a metric $\rho$ that compares two such outcomes. At each sampled environment step $u$, this reference and the action $a^{\rm exec}_{i,u}$ the agent executed are both evaluated at the state $s_u$ it actually occupied, giving the per-step behavioral residual $d_{i,u}$ and its episode total $D_i$,
\begin{equation}
d_{i,u}=\rho\!\left(\phi_i(s_u,a^{\rm exec}_{i,u}),\,\phi_i(s_u,a^{*}_{i,u})\right),
\quad D_i=\sum_{u\in\mathcal T}d_{i,u},
\label{eq:anomaly-score}
\end{equation}
where $\mathcal T=\{h,2h,\ldots,h\lfloor L/h\rfloor\}$ collects the sampling instants of an episode of length $L$. Any $\phi_i$ whose outcomes are recorded rather than asserted admits this construction, since a verifier then recomputes $d_{i,u}$ by arithmetic on committed values without re-executing the policy, and a bounded $\rho$ common to all agents keeps $D_i$ finite and comparable across the population. Residuals measure inconsistency with the reference rather than malicious intent, since an agent that samples its action, or whose outcome the environment alters, differs from the reference while acting honestly. Honest residuals therefore need not concentrate near zero.

On retained agents, define the population mean $\mu_t$ and population standard deviation $\sigma_t$ of $D_i$. The population-relative detector raises a deviation flag $f_i(t)$ when a retained agent's cumulative residual exceeds the threshold $\tau_t$,
\begin{equation}
\tau_t=\mu_t+k\sigma_t,\qquad
f_i(t)=\mathbf1\{i\in\mathcal K_t,\ D_i>\tau_t\},
\label{eq:threshold}
\end{equation}
with sensitivity $k>0$. The comparison is strict, so threshold ties and zero variance produce no flags. If $\mathcal T$ is empty, the analytical recursion skips that episode. With no retained agents, all weights are zero.

Since $\mu_t$ and $\sigma_t$ include the adversaries, deviating agents inflate the reference they are judged against, and this self-masking caps how many can be flagged at one settlement.

\begin{proposition}[Detector capacity]
\label{prop:capacity}
For $k>0$ and $n\ge1$, the maximum simultaneous flag count is
\begin{equation}
F_{\max}(n)=\left\lceil\frac{n}{1+k^2}\right\rceil-1.
\label{eq:capacity}
\end{equation}
Consequently, with $M_t>0$ adversaries still active, the realized flagged-adversary fraction is at most $\min\{1,F_{\max}(n)/M_t\}$. Simultaneous coverage of all $M_t$ requires $k<\sqrt{n/M_t-1}$.
\end{proposition}
\begin{proof}
The idea is that flagged agents raise the mean and spread they are judged against, so only a limited number can exceed $\mu+k\sigma$ at once. Fix a settlement, drop the index $t$, and write $x_i=D_i-\mu$ for the centered residuals of the $n$ retained agents, so that $\sum_ix_i=0$ and $\sum_ix_i^2=n\sigma^2$.

\emph{Upper bound.} Let $F\ge1$ agents be flagged, so each has $x_i>k\sigma$. Then $\sigma>0$, because $\sigma=0$ makes every $x_i$ zero. The flagged residuals add up to some $a>Fk\sigma$, so the remaining $n-F$ add up to $-a$, and hence $F<n$. By Cauchy--Schwarz, $m$ numbers with sum $A$ have squares summing to at least $A^2/m$. Applying this to both groups,
\[
n\sigma^2\ge\frac{a^2}{F}+\frac{a^2}{n-F}=\frac{a^2n}{F(n-F)}>\frac{k^2\sigma^2Fn}{n-F}.
\]
Dividing by $n\sigma^2$ and multiplying by $n-F$ gives $n-F>k^2F$, that is, $F<n/(1+k^2)$. The largest integer strictly below $n/(1+k^2)$ is $F_{\max}(n)$.

\emph{The bound is reached.} For any $F<n/(1+k^2)$, give $F$ agents residual $\delta>0$ and the rest zero. Then $\sigma=\delta\sqrt{F(n-F)}/n$, and the high agents' centered residual $(n-F)\delta/n$ exceeds $k\sigma$ exactly when $n-F>k^2F$, so exactly $F$ agents are flagged. Equal residuals give $F=0$.

\emph{Adversary coverage.} Flagged adversaries are among the flagged agents, so at most $F_{\max}(n)$ of the $M_t$ are flagged. Flagging all of them needs $F_{\max}(n)\ge M_t$, which holds exactly when $n/(1+k^2)>M_t$, that is, $k<\sqrt{n/M_t-1}$.
\end{proof}

\input{revision_generated/detector_capacity_table.tex}

Table~\ref{tab:detector-capacity} gives the initial ceilings. At $n=8$ at most two agents can be flagged, bounding the true-positive fraction by $2/3$ or $2/5$ for three or five adversaries, and no flags are possible once $n\le3$. At the $5/8$ ratio the ceiling stays below $1/2$ for every evaluated $N$, so a majority adversary population is never more than half flagged in one round and separation is not guaranteed. The verifier needs only residuals and eligibility, not ground-truth labels, and it reads $D_i$ only through its standing in the population, as do the token and aggregation rules, so any admissible $(\phi_i,\rho)$ leaves them unchanged.

\subsection{Token-Based Incentive Mechanism}
\label{sec:incentive}

A retained agent self-funds stake $s$, forfeiting it when flagged. Unflagged agents recover their principal and receive reward $\beta s$ plus equal compensation from the fraction $\lambda$ of slashed deposits. With $F_t=\sum_{i\in\mathcal K_t}f_i(t)$ flags, the per-recipient compensation is
\[
c_t=\frac{\lambda sF_t}{n_t-F_t}.
\]
The fraction $1-\lambda$ of slashed stake is burned; rewards are minted. With integer tokens, rounding compensation downward and burning the residue preserves the upper bounds below. For a retained agent $i\in\mathcal K_t$ in an evidence-bearing settlement,
\begin{equation}
S_i(t+1)=S_i(t)+\begin{cases}\beta s+c_t,&f_i(t)=0,\\-s,&f_i(t)=1.\end{cases}
\label{eq:token-update}
\end{equation}
An agent is permanently excluded precisely when $S_i<S_{\min}$; equality remains eligible. Excluded balances are frozen, and their aggregation weights are zero. Economic exclusion concerns training influence; it does not itself remove a physical agent. Section~\ref{sec:experimental-setup} gives the evaluation parameters; the following theorem separately states the conditions sufficient for a continued, self-funded staking process.

\begin{lemma}[Conditional one-round drift]
\label{lem:drift}
At an evidence-bearing settlement, given pre-settlement history $\mathcal F_t$, let $p_i(t)$ be the conditional flag probability and $b_i(t)=\beta s+\mathbb E[c_t\mid\mathcal F_t,f_i(t)=0]$. Then
\begin{equation}
\mathbb E[\Delta S_i\mid\mathcal F_t]=(1-p_i(t))b_i(t)-p_i(t)s.
\label{eq:drift}
\end{equation}
When the unflagged event has probability zero, set $b_i(t)=\beta s$; its weighted gain term is then zero. The expression is nonnegative exactly when $p_i(t)\le b_i(t)/(s+b_i(t))$.
\end{lemma}
\begin{proof}
By~\eqref{eq:token-update}, a retained balance changes in only two ways: it gains $\beta s+c_t$ when unflagged and loses $s$ when flagged. The loss occurs with conditional probability $p_i(t)$ and contributes $-p_i(t)s$. The gain occurs with probability $1-p_i(t)$, and its average size given that the agent is unflagged is $b_i(t)$, so it contributes $(1-p_i(t))b_i(t)$. Adding the two gives~\eqref{eq:drift}. If $p_i(t)=1$, the gain never occurs, and the convention $b_i(t)=\beta s$ keeps its term at zero. The compensation $c_t$ depends on every agent's flag, but it enters only through this conditional average, so no independence between agents is needed. Finally, since $s+b_i(t)>0$, the drift is nonnegative exactly when $b_i(t)\ge p_i(t)(s+b_i(t))$, that is, $p_i(t)\le b_i(t)/(s+b_i(t))$.
\end{proof}

\begin{theorem}[One-round incentive compatibility]
\label{thm:incentive-compatibility}
At a fixed history $\mathcal F_t$, compare honest execution and a specified deviation. The subscripts $H$ and $D$ evaluate the conditional flag probability $p_i(t)$ and conditional unflagged gain $b_i(t)$ of Lemma~\ref{lem:drift} under honest execution and under the deviation, respectively, and $u_{\rm dev}\ge0$ is the additional deviation utility in token-equivalent units. For an agent maximizing expected one-round utility, honest execution strictly dominates this deviation if and only if
\begin{equation}
(p_D-p_H)s+(1-p_H)b_H-(1-p_D)b_D>u_{\rm dev}.
\label{eq:ic-condition}
\end{equation}
\end{theorem}
\begin{proof}
By Lemma~\ref{lem:drift}, the expected one-round token change is $(1-p_H)b_H-p_Hs$ under honest execution and $(1-p_D)b_D-p_Ds$ under the deviation, which additionally earns $u_{\rm dev}$. Honest execution is strictly better exactly when
\[
(1-p_H)b_H-p_Hs>(1-p_D)b_D-p_Ds+u_{\rm dev}.
\]
Grouping the two deposit terms as $(p_D-p_H)s$ turns this inequality into~\eqref{eq:ic-condition}. Each step can be reversed, so the condition is both necessary and sufficient.
\end{proof}
\begin{corollary}[Common-gain incentive margin]
\label{cor:incentive-margin}
If additionally $b_H=b_D=b$, the strict-deterrence boundary is
\begin{equation}
\Delta=(p_D-p_H)(b+s),\qquad u_{\rm dev}<\Delta.
\label{eq:incentive-margin}
\end{equation}
For $u_{\rm dev}=0$, this reduces to $p_D>p_H$.
\end{corollary}
\begin{proof}
When both actions earn the same unflagged gain $b$, the two gain terms of~\eqref{eq:ic-condition} differ only through their probabilities: $(1-p_H)b-(1-p_D)b=(p_D-p_H)b$. The left side of~\eqref{eq:ic-condition} then equals $(p_D-p_H)s+(p_D-p_H)b=\Delta$, so the condition reads $u_{\rm dev}<\Delta$. With $u_{\rm dev}=0$, it requires $(p_D-p_H)(s+b)>0$, which holds exactly when $p_D>p_H$ because $s+b>0$.
\end{proof}
The common-gain condition is substantive: changing one trajectory can alter peer flags and compensation. These are one-round comparisons, not a dynamic equilibrium result. If $u_{\rm dev}\ge\Delta$, the specified deviation is not strictly deterred at that operating point.

\begin{auss}[Conditions for exclusion]
\label{as:exclusion}
An agent starts retained, re-stakes from its own balance without top-ups, and is permanently excluded below $S_{\min}\ge s$. Evidence-bearing settlements continue while it remains retained.
\end{auss}
For all possible active populations, capacity supplies the pointwise bounds
\begin{equation}
\begin{aligned}
\kappa_N&=\max_{1\le n\le N}\frac{F_{\max}(n)}{n-F_{\max}(n)},\\
c_t&\le C_{\max}=\lambda s\kappa_N,\qquad b_{\max}=\beta s+C_{\max}.
\end{aligned}
\label{eq:compensation-bound}
\end{equation}
Indeed, more flags mean more slashed stake shared among fewer recipients, so $c_t$ grows with $F_t$; Proposition~\ref{prop:capacity} caps $F_t$ at $F_{\max}(n_t)$, and the population only shrinks from $N$, so every realized $c_t$ is at most $C_{\max}$. Because this holds in every outcome, it also holds on average under any conditioning, giving $b_i(t)\le b_{\max}$ with no independence or decreasing-compensation assumption.

\begin{theorem}[Finite expected exclusion time]
\label{thm:divergence}
Under Assumption~\ref{as:exclusion}, let $\underline p$ be the infimum of the conditional flag probabilities $p_j(t)$ of Lemma~\ref{lem:drift} over all evidence-bearing settlements and all pre-settlement histories before exclusion. If $\underline p>P^*=b_{\max}/(s+b_{\max})$, the number $T_j^{\rm ex}$ of evidence-bearing settlements until agent $j$ is excluded satisfies
\begin{equation}
\mathbb E[T_j^{\rm ex}]\le
\frac{S_j(0)-S_{\min}+s}{\varepsilon},\quad
\varepsilon=\underline p s-(1-\underline p)b_{\max}>0.
\label{eq:exclusion-time}
\end{equation}
\end{theorem}
\begin{proof}
The idea is that agent $j$'s balance falls by at least $\varepsilon$ per settlement on average and cannot fall far below $S_{\min}$, so the expected number of settlements before exclusion is at most its distance to that floor divided by $\varepsilon$. Count only evidence-bearing settlements, write $T=T_j^{\rm ex}$, and let $\tilde S_j(t)=S_j(t\wedge T)$ be the balance frozen at exclusion.

\emph{Step 1 (average loss per settlement).} The drift $(1-p)b-ps$ in~\eqref{eq:drift} rises with the gain $b$ and falls with the flag probability $p$. While $j$ is retained, $b_j(t)\le b_{\max}$ by~\eqref{eq:compensation-bound} and $p_j(t)\ge\underline p$, so
\[
\mathbb E[\Delta S_j\mid\mathcal F_t]\le(1-\underline p)b_{\max}-\underline p s=-\varepsilon,
\]
and $\varepsilon>0$ is the same as $\underline p(s+b_{\max})>b_{\max}$, that is, $\underline p>P^*$.

\emph{Step 2 (adding up $r$ settlements).} The frozen balance is the initial balance plus the changes made while $j$ was retained, $\tilde S_j(r)=S_j(0)+\sum_{t<r}\mathbf 1\{T>t\}\Delta S_j(t)$. Whether $T>t$ is known before settlement $t$, so Step 1 gives $\mathbb E[\mathbf 1\{T>t\}\Delta S_j(t)]\le-\varepsilon\Pr[T>t]$. Since $\sum_{t<r}\Pr[T>t]=\mathbb E[T\wedge r]$,
\[
\mathbb E[\tilde S_j(r)]\le S_j(0)-\varepsilon\,\mathbb E[T\wedge r].
\]
These expectations are finite because every change lies in $[-s,b_{\max}]$.

\emph{Step 3 (the balance has a floor).} A retained balance is at least $S_{\min}$, and the settlement that excludes $j$ removes at most $s$, so $\tilde S_j(r)\ge S_{\min}-s$. With Step 2, $\varepsilon\,\mathbb E[T\wedge r]\le S_j(0)-S_{\min}+s$ for every $r$. As $r\to\infty$, $\mathbb E[T\wedge r]$ increases to $\mathbb E[T]$, which gives~\eqref{eq:exclusion-time}; in particular, $T$ is finite with probability one.
\end{proof}

\begin{corollary}[Exclusion of new gradient contributions]
\label{cor:exclusion}
If the theorem's conditions hold for every adversary in the finite initial population, all are excluded in finite expected time and thereafter supply zero direct gradient contribution. Updates use retained honest contributors if any remain.
\end{corollary}
\begin{proof}
Theorem~\ref{thm:divergence} gives each adversary a finite expected exclusion time, and with finitely many adversaries the last exclusion also has finite expectation, since $\mathbb E[\max_jT_j^{\rm ex}]\le\sum_j\mathbb E[T_j^{\rm ex}]$. Exclusion is permanent, so an excluded agent never rejoins the retained set: it receives weight $\omega_j=0$ in~\eqref{eq:aggregation-weight} of Section~\ref{sec:aggregation}, and its gradient term vanishes from $G_{t+1}$ in~\eqref{eq:robust-update}. After the last exclusion, only honest agents carry weight, and the update is skipped if none remain.
\end{proof}

Self-funding and exclusion are accounting requirements, whereas the history-uniform detection floor is a condition on the verifier and environment: a surviving adversary at $n\le3$, for example, contradicts a positive floor, so pooled detection rates alone cannot certify~\eqref{eq:exclusion-time}. The clock counts evidence-bearing settlements only. For an honest agent, a false-flag bound $p_H(t)\le\beta/(1+\beta)$ gives nonnegative drift but does not guarantee retention, since positive-drift processes can still cross the threshold. Section~\ref{sec:theory-evidence} reads these conditions against the evaluated configuration.

For $N=8$, $\kappa_N=0.4$ (at $n=7,F=2$), so $b_{\max}=(0.1+0.4\lambda)s$ for $\beta=0.1$, and the sufficient thresholds $P^*$ are $1/11$, $3/13$, and $1/3$ for $\lambda=0,0.5,1$. Taken together, Theorem~\ref{thm:divergence} and these values show that exclusion does not require a perfect detector: even with full compensation, $P^*$ is at most $1/3$ and holds for every retained population, so under Assumption~\ref{as:exclusion} an adversary whose conditional flag probability never drops below some $\underline p>1/3$ is excluded in finite expected time, after which Corollary~\ref{cor:exclusion} removes its direct gradient contribution.

\subsection{Token-Weighted Policy Aggregation}
\label{sec:aggregation}

After settlement $t$, normalize balances only over the new retained set $\mathcal K_{t+1}$:
\begin{equation}
\omega_i(t+1)=\begin{cases}
S_i(t+1)/\sum_{j\in\mathcal K_{t+1}}S_j(t+1),&i\in\mathcal K_{t+1},\\
0,&\text{otherwise}.
\end{cases}
\label{eq:aggregation-weight}
\end{equation}
Weights are held fixed during each learning update. To state a bounded-influence result, suppose each per-agent loss gradient $g_i=\nabla_\theta\ell_i$ is clipped to a prescribed norm limit $L_g>0$:
\begin{equation}
\bar g_i=g_i\min\!\left\{1,\frac{L_g}{\|g_i\|_2}\right\},\qquad
G_{t+1}=\sum_i\omega_i(t+1)\bar g_i.
\label{eq:robust-update}
\end{equation}
Take $\bar g_i=0$ when $g_i=0$. The optimizer consumes $G_{t+1}$; for adaptive optimizers, its parameter increment need not be a learning rate times this vector. Clipping makes every $\|\bar g_i\|_2\le L_g$, so by the triangle inequality the direct pre-optimizer contribution of any contributor set $\mathcal B$ is at most $L_g$ times its total weight:
\begin{equation}
\left\|\sum_{i\in\mathcal B}\omega_i(t+1)\bar g_i\right\|_2
\le L_g\sum_{i\in\mathcal B}\omega_i(t+1).
\label{eq:honest-update-limit}
\end{equation}
Exclusion makes an agent's coefficient exactly zero, while prior influence in model parameters or optimizer state can persist. If no contributors remain, the update is skipped. The norm bound is a sufficient design condition; the evaluation does not establish a numerical per-agent gradient bound.

\section{Implementation and Blockchain Integration}
\label{sec:implementation}

\subsection{Smart Contract Execution Pipeline}
\label{sec:contract-pipeline}

The architecture separates the blockchain coordination layer from the SCRIMP learner. The contracts correspond to the stages in Fig.~\ref{fig:workflow} and provide persistent records at their interfaces. \texttt{BatchDataStorage} receives sampled interaction records $(o_{i,t},a_{i,t},r_{i,t},m_{i,t-1})$, associated with the agent, episode and policy version. The \texttt{OutputHead}, \texttt{TDError}, \texttt{TrainingGAE} and \texttt{PolicyLoss} interfaces anchor policy/value outputs, temporal-difference residuals, advantages and policy-update information, respectively. Their names identify the computations being recorded; neural inference and optimization remain off-chain.

At episode end, \texttt{MaliciousDetection} uses recommended and executed behavior to form the cumulative residuals and population-relative flags of Section~\ref{sec:detection}. \texttt{IncentiveMechanism} consumes those flags, applies deposit forfeiture, rewards, compensation and burning, and updates balances and eligibility. This state determines the weights for the subsequent policy update. Keeping the detection result and its economic consequence in the same episode sequence prevents an update from combining flags from one round with balances from another. Module boundaries organize access to shared state; they do not, by themselves, establish a gas advantage over a single-contract implementation.

\subsection{On-Chain and Off-Chain Function Allocation}
\label{sec:on-off-chain}

Observation encoding, Transformer communication, the recurrent output head, temporal-difference and generalized advantage estimation, PPO losses, gradient computation and the final token-weighted update run in the learner. General neural computation is unsuitable for direct execution in the EVM~\cite{wohrer2018smart}. The coordination layer instead records the behavioral evidence, detector outputs, economic state and policy commitments needed to associate an update with its inputs.

A policy hash identifies a particular parameter snapshot, while interaction commitments identify the records used for an episode. This supports checking whether two participants refer to the same data and model. A hash is not a proof that the committed model produced an action, that a submitted loss or gradient was calculated correctly, or that physical telemetry is truthful. Such checks require the corresponding records and an agreed observation source. Likewise, a signed submission authenticates its sender without establishing the accuracy of its contents. These distinctions separate tamper-evident coordination from behavioral verification.

Tuple submission can overlap off-chain computation, but the learner must use a consistent policy version and economic state. An update that depends on an unsettled balance transition must wait for the required state; asynchronous submission alone does not establish an end-to-end latency bound.

\subsection{Solidity Implementation and Geth Architecture}
\label{sec:geth-deployment}

The software stack comprises Solidity~0.8.0 contracts, Geth~1.13.14 execution nodes, and a Python~3.9.21 learner connected through Web3.py~7.4.0 and HTTP-RPC. Truffle provides contract compilation/deployment interfaces and the application binary interface used for submissions and state queries~\cite{dhanvardini2023development}. The intended placement is a private Ethereum network coordinating authorized participants. Execution-client configuration, validator operation and the chosen finality policy are separate deployment concerns; the module description does not establish consensus performance.

Contract arithmetic represents numerical quantities using scale $10^6$, giving six decimal places for scaled values. Detector decisions must preserve the strict comparison in Section~\ref{sec:detection}, including ties. For $q$ eligible agents, let $Q=\sum_iD_i$, $B=\sum_iD_i^2$ and $x_i=qD_i-Q$, with both sums over eligible agents. Squaring clears the standard deviation from the test of Section~\ref{sec:detection}, which at $k=1.5$ becomes the division-free comparison $x_i>0$ and $4x_i^2>9(qB-Q^2)$, avoiding a rounded mean or standard deviation in the decision. The same eligibility set defines the balance sum used for normalization. Token arithmetic must account for division residue explicitly so that distributed compensation and burned amounts reconcile with the slashed deposits.

\subsection{Workload and IoT Placement}
\label{sec:system-overhead}

With $K$ sampling boundaries and $N$ agents, interaction collection produces $NK$ agent records. At a 32-step interval and a 256-step episode limit, $K\leq8$. Under one submission per sampled record, collection plus separate detection and incentive calls requires at most $8N+2$ requests for those operations. This count excludes deployment, deposit handling, additional policy-stage commitments and retries; it is not a measurement of the complete transaction workload. Record volume scales with $NK$, and balance accounting scans the participating agents. Gas consumption, block capacity and communication delay additionally depend on contract storage, payloads and network configuration.

For a multi-operator automated guided vehicle fleet, coordinator nodes can host the learner and blockchain interfaces while vehicles execute the transferred policy locally. Recommended and measured vehicle positions supply the behavioral comparison. The chain coordinates training and its economic consequences, rather than serving each navigation decision. Deployment suitability therefore depends on both trustworthy telemetry and device-specific control requirements; neither an edge-control timing budget nor a network-throughput claim follows from the MAPF task results.

\section{Performance Evaluation}\label{sec:performance-evaluation}

\subsection{Experimental Setup}\label{sec:experimental-setup}

\emph{Task.} MAPF represents coordinated navigation in applications such as automated guided vehicle fleets, intelligent transportation and drone swarms~\cite{alkazzi2024comprehensive,zong2022mapdp}. We instantiate BE-CTDE on the SCRIMP backbone~\cite{wang2023scrimp}: agents choose among four cardinal movements and staying, using local obstacle, goal and peer observations with recurrent Transformer-based communication. This task fixes the outcome map of Section~\ref{sec:detection} to the successor grid cell under Euclidean distance, with the most probable feasible move as the reference action. The residual is then $d_{i,u}=\|p^{\rm exec}_{i,u}-p^{*}_{i,u}\|_2$, the distance in grid cells between the executed cell $p^{\rm exec}_{i,u}$ and the reference cell $p^{*}_{i,u}$, leaving the mechanism unchanged from its general form. Rewards penalize collisions and idling and reward goal achievement and exploration. The overt adversarial behaviors follow the AMAPF setting~\cite{ivanova2021adversarial}: static blocking, occupation of goal regions, and stagnant or repetitive movement. Their effect on the residual depends on the realized trajectory and the active population; persistence alone does not establish detectability.

\emph{Setup.} Training uses eight agents on a $40\times40$ grid with randomly assigned unique starts and goals, at most $L=256$ steps per episode, and triangular obstacle density with minimum zero, mode 33\% and maximum 50\%. Training runs for 1,838 episodes. Residuals are sampled every $h=32$ steps, giving at most eight samples per episode, and are settled at episode end with sensitivity $k=1.5$. Testing covers $N\in\{8,32,64,128\}$ and obstacle densities 0\%, 15\% and 30\%, with 100 episodes per cell. Larger populations transfer the $N=8$ policy without retraining. The computing platform is an Intel Core i7-13700 with 32~GB memory and an RTX~3090~Ti.

Table~\ref{tab:simulation-params} summarizes the MAPF setup. We examine two test conditions: adversaries remain present in the adversarial-testing comparison, while all agents behave honestly in the compensation study. Both conditions use the same backbone, training population, map family and evaluation grid.

\begin{table}[t]
\caption{Simulation Parameters}
\label{tab:simulation-params}
\centering\footnotesize
\setlength{\tabcolsep}{4pt}
\renewcommand{\arraystretch}{1.04}
\begin{tabular}{@{}>{\raggedright\arraybackslash}p{.42\columnwidth}>{\raggedright\arraybackslash}p{.53\columnwidth}@{}}
\toprule
Parameter & Value \\
\midrule
\multicolumn{2}{@{}l}{\emph{Training environment}} \\
Grid; number of agents & $40\times40$; 8 \\
Start and goal assignment & Random unique positions \\
Training adversaries & 3 or 5 of 8 agents \\
Obstacle density & Triangular $(0,0.33,0.50)$ \\
Episodes; maximum steps & 1,838; 256 per episode \\
\midrule
\multicolumn{2}{@{}l}{\emph{Detection and token mechanism}} \\
Residual sampling & Every 32 steps \\
Detection and settlement & At episode end \\
Sensitivity; reward factor & $k=1.5$; $\beta=0.1$ \\
Deposit; exclusion threshold & $s=1$; $S_{\min}=0.1s$ \\
Compensation factor & $\lambda\in\{0,0.5,1\}$ \\
\midrule
\multicolumn{2}{@{}l}{\emph{Evaluation environment}} \\
Grid; agent populations & $40\times40$; 8, 32, 64, 128 \\
Obstacle densities & 0\%, 15\%, 30\% \\
Episodes per condition & 100 \\
Adversarial testing & 37.5\% or 62.5\% adversaries \\
Honest testing & All agents follow the policy \\
\bottomrule
\end{tabular}
\end{table}

\emph{Metrics and baselines.} We report collision rate (CO) and success rate (SR). CO is defined as honest-agent obstacle/boundary collisions per honest-agent time step; agent--agent contacts are separate. SR is the percentage of test episodes in which all honest agents reach their goals within the time limit:
\begin{equation}
\mathrm{CO}=100\frac{N_{\mathrm{collision}}}{\sum_e T_eN_{\mathrm{honest},e}},\qquad
\mathrm{SR}=100\frac{N_{\mathrm{success}}}{E}.
\label{eq:success-rate}
\end{equation}
CO and SR are interpreted jointly: low collision frequency alone cannot distinguish safe progress from stalled navigation. Training goal arrivals provide additional context for learned behavior.

The baselines add the elements of the closed loop one at a time on a common backbone, and several of them instantiate the defining mechanism of a closely related design. \emph{SCRIMP}~\cite{wang2023scrimp}, a state-of-the-art learning-based MAPF method and the backbone of BE-CTDE, is trained without any defense. \emph{On-Chain CTDE} adds the blockchain coordination pipeline without detection or incentives; as in LearnChain~\cite{sami2024learnchain}, the chain records training but leaves every agent's influence unchanged. \emph{BE-CTDE Lite} adds the detector of Section~\ref{sec:detection} without the token economy, so flags are recorded but flagged agents still contribute equally. Token-enabled BE-CTDE adds the token economy of Section~\ref{sec:incentive} and balance-weighted aggregation, with $\lambda\in\{0,0.5,1\}$. It shares contract-automated penalties and rewards with Tian~\cite{tian2025blockchain}, and its $\lambda=0$ variant burns all slashed stake as in burn-based incentive designs~\cite{zhao2025blockchain}; unlike both, it also turns balances into aggregation weights. Instantiating these mechanisms on one backbone, rather than reproducing each full algorithm, keeps every comparison within a single training and evaluation protocol.

\subsection{Performance Under Adversarial Training and Testing}\label{sec:results-adversarial}
The adversarial-testing condition places adversaries in both training and testing, using malicious-agent ratios of 37.5\% and 62.5\%. Metrics concern honest agents. The larger-$N$ rows are zero-shot population transfers, and test-time policy execution requires no chain interaction. Table~\ref{tab:results-adversarial} reports CO and SR.

\input{revision_generated/restored_original_adversarial_table.tex}

\emph{Minority adversaries.} Full BE-CTDE achieves 100\% SR at $N=8$ across all densities, compared with 96--97\% for Lite and zero for SCRIMP. At $N=32$ and 30\% obstacles, its CO is 23.09\% versus Lite's 27.88\%, a 4.79-percentage-point (pp) reduction; SR rises from 92\% to 97\%. At $N=128$, increasing congestion in the fixed grid reduces SR, but full BE-CTDE retains the advantage over Lite.

\emph{Majority adversaries.} Full BE-CTDE achieves SR of at least 93\% at $N\le32$. At $N=64$ and 15\% obstacles, CO falls from Lite's 37.10\% to 24.00\%, the largest reduction at 13.10~pp (35.3\%), while SR increases from 89\% to 95\%. Compared with SCRIMP, the relative CO reductions are 66.4\% at $N=32$, 15\% obstacles, and 81.0\% at $N=64$, zero obstacles. These comparisons describe task performance; the contrasts do not isolate every implementation difference between methods.

Full BE-CTDE has lower CO and higher SR than Lite in all 24 ratio/population/density cells. The largest SR difference is 7~pp. This consistency across test conditions is useful evidence of the task-level benefit, but the cells share trained policies and are not independent training replications. Under binomial sampling of 100 test episodes, Wilson 95\% intervals have approximate half-widths of 3.7~pp at 97\% SR and 5.2~pp at 93\% SR; these describe individual success proportions, not a confidence interval for a paired method difference. The available experiment records do not establish independent training-seed variation.

\subsection{Effect of Compensation Factor}\label{sec:results-compensation}
The compensation factor acts only through training: it changes balances and, through them, the aggregation weights of~\eqref{eq:aggregation-weight} and the timing of exclusion, whereas test-time execution involves no chain interaction. Its effect therefore resides entirely in the learned policy. Adversaries present at test time would confound this effect with interference that no training-time mechanism removes, since economic exclusion (Section~\ref{sec:incentive}) withdraws training influence without removing the agent itself. The honest-testing condition removes this interference by design: adversaries take part only in training, and every test agent follows the transferred policy. Varying $\lambda$ under this condition, within the same MAPF evaluation protocol, examines the compensation effect on the learned policy and complements the adversarial-testing comparison of Section~\ref{sec:results-adversarial}, which retains test-time adversaries. Figure~\ref{fig:m3_performance} shows policy-loss and goal-arrival traces for three adversaries during training.

\begin{figure}[!tb]
\centering
\subfloat[Policy loss]{\includegraphics[width=.49\columnwidth]{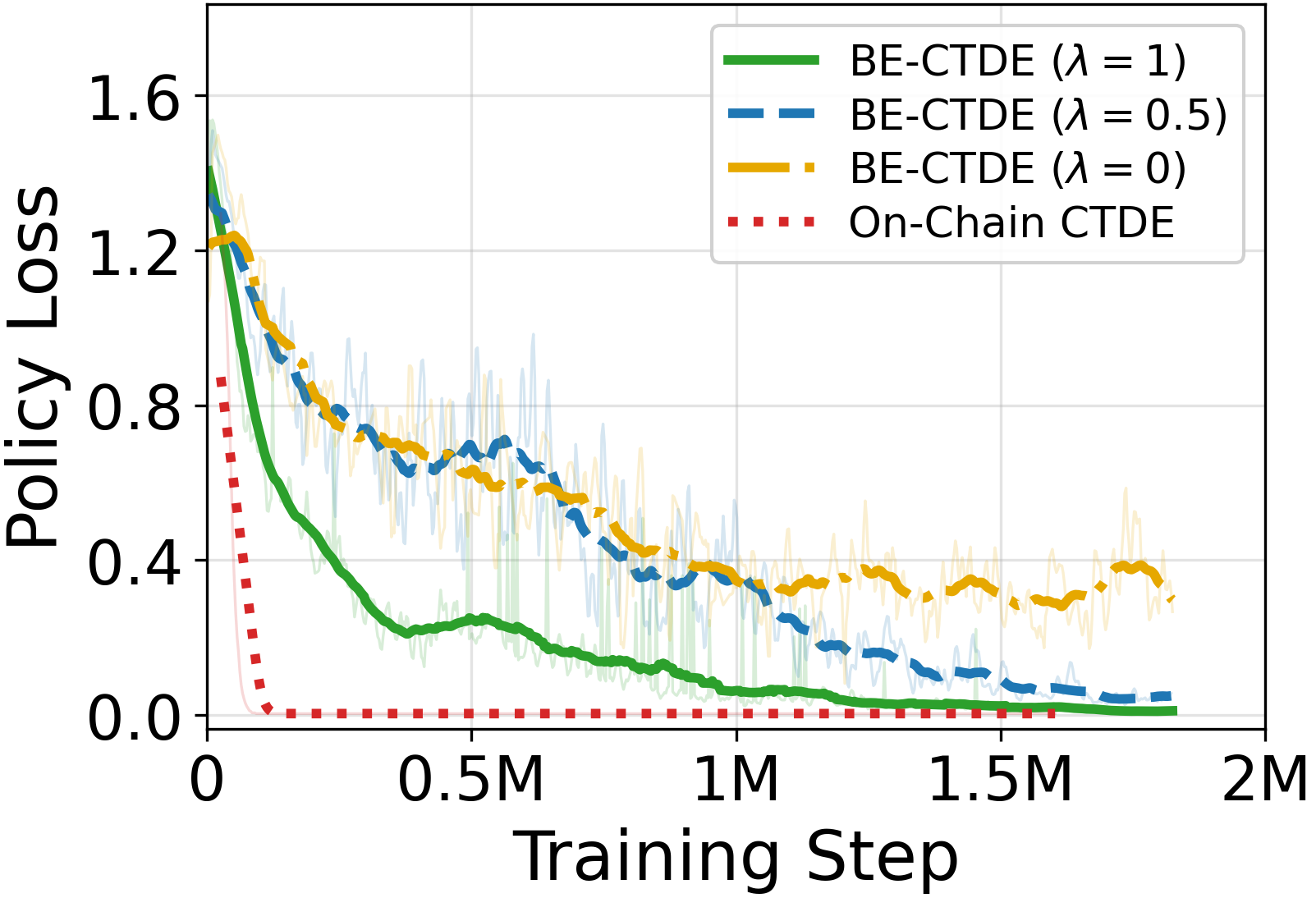}\label{fig:m3_policyloss}}\hfill
\subfloat[Goal arrivals]{\includegraphics[width=.49\columnwidth]{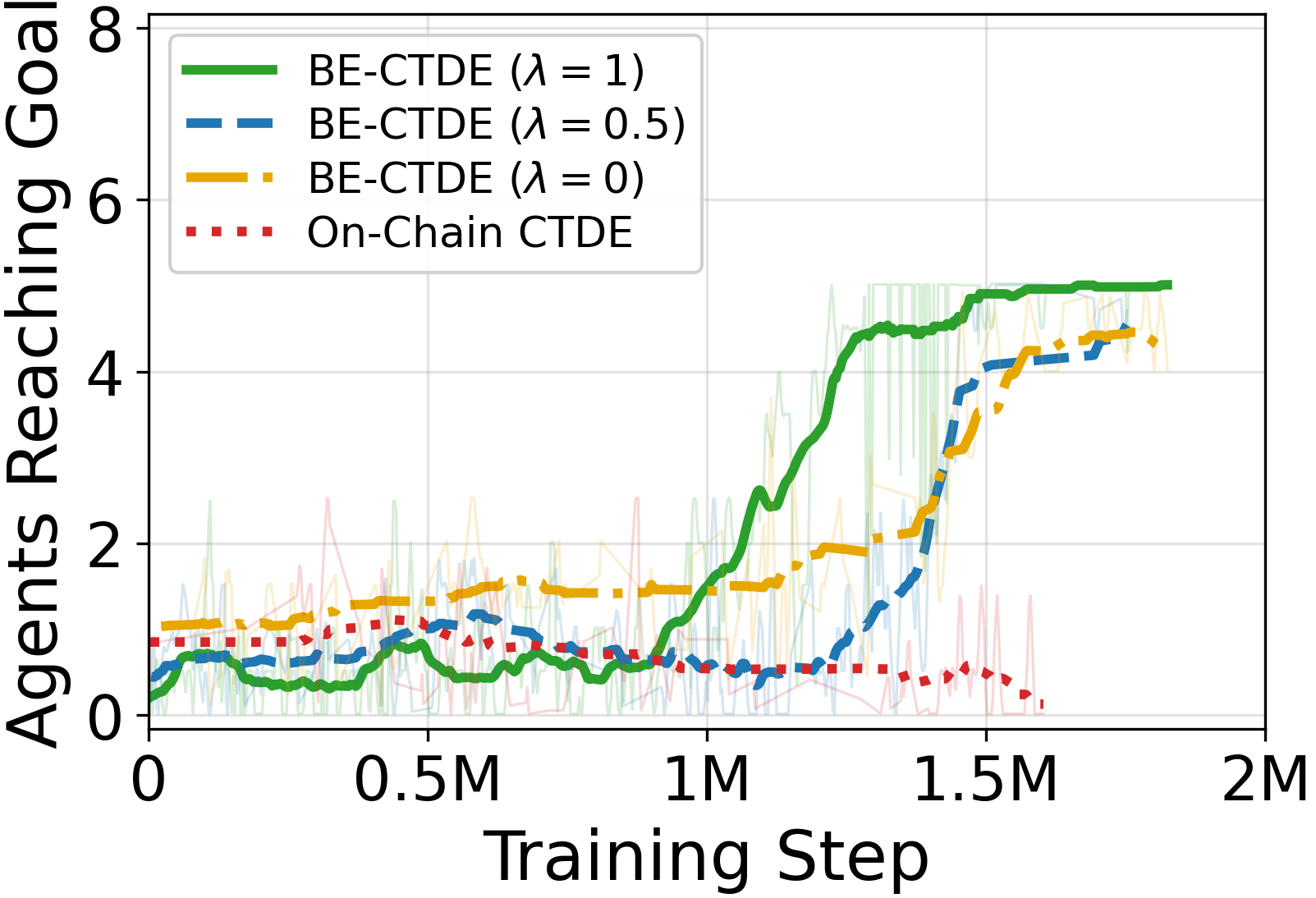}\label{fig:m3_goal}}
\caption{Policy-loss and goal-arrival curves with three malicious agents during training. The horizontal axis is the recorded training-step counter; the curves measure learning progress, not exclusion events.}
\label{fig:m3_performance}
\end{figure}

The loss curves show faster decline for the compensated variants, while the goal-arrival curves separate task progress from loss convergence alone. In particular, arrivals approaching five are consistent with the five honest agents learning useful behavior, but do not establish when, or whether, the three malicious contributors were excluded. Learning to navigate interference can also change arrivals.

\input{revision_generated/restored_original_compensation_table.tex}

Table~\ref{tab:results-compensation} reports CO under honest testing. At $N=8$, 30\% obstacles and a 37.5\% training-time malicious ratio, full compensation reduces CO by 8.84~pp relative to burn-only. The largest relative reduction is 43.1\%, from 20.38\% to 11.59\% at $N=32$, 15\% obstacles and the same ratio. It improves on burn-only in all 24 cells and attains the lowest CO among the three token variants in 23/24; the exception is $N=32$, 62.5\% training adversaries and zero obstacles, where $\lambda=.5$ gives 14.86\% versus 15.19\% for $\lambda=1$. This distinction avoids conflating the pairwise burn-only comparison with the three-variant ranking.

\begin{figure}[!tb]
\centering
\includegraphics[width=\columnwidth,trim=7bp 7bp 7bp 7bp,clip]{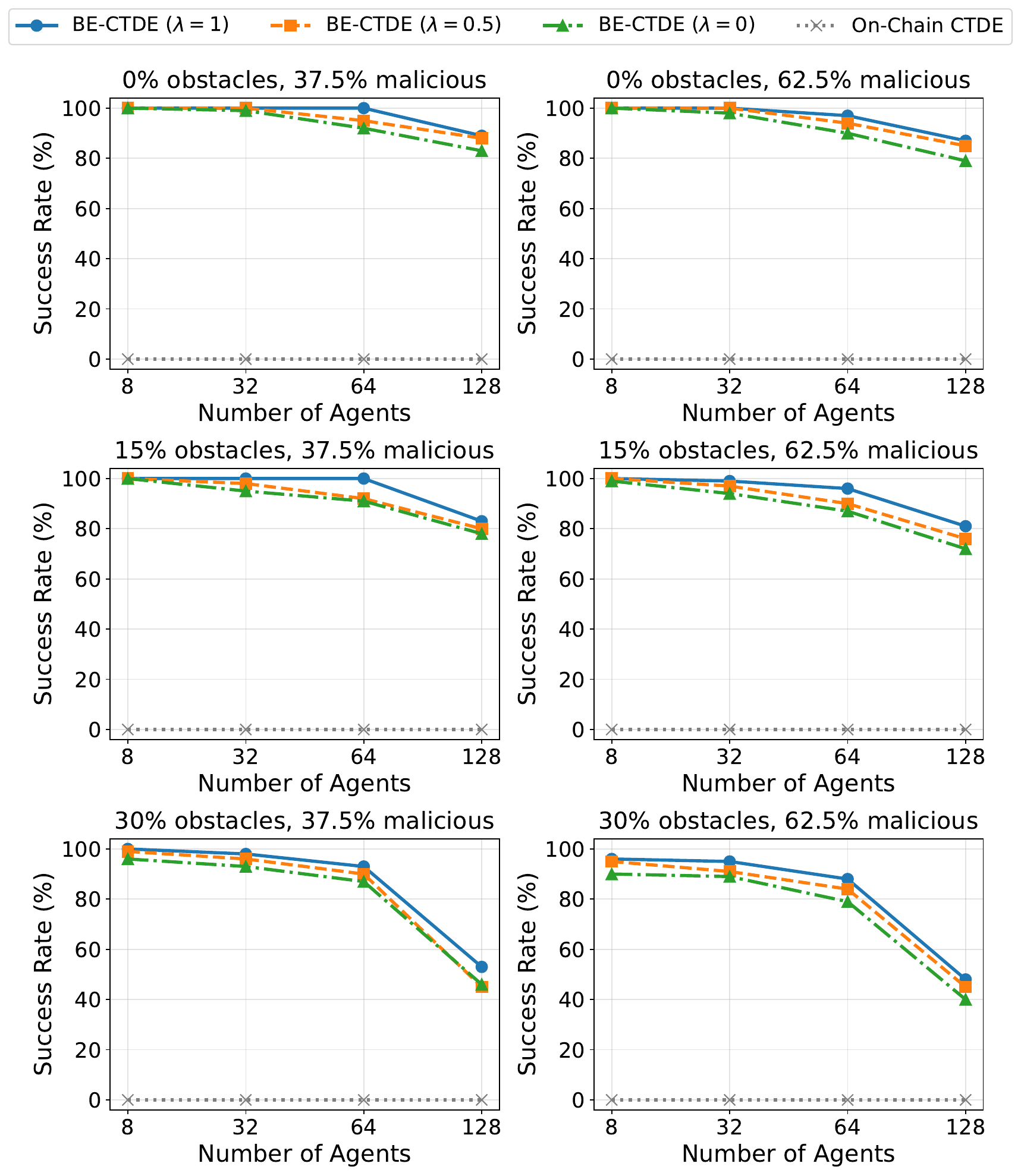}
\caption{SR comparison with adversaries during training and honest testing. Rows show obstacle densities 0\%, 15\% and 30\%; columns show training-time malicious ratios 37.5\% and 62.5\%.}
\label{fig:sr_comparison}
\end{figure}

Figure~\ref{fig:sr_comparison} reports the corresponding SR comparison. It shows the compensation advantage alongside the decline at larger populations and higher obstacle density. The honest-testing and joint-adversarial-testing results address different transfer conditions and should not be combined into a single matched ablation. Likewise, low On-Chain CTDE collision values at small populations should be interpreted with its goal-arrival behavior, without treating passivity as demonstrated solely by a collision rate.

\subsection{Relation to the Analysis}\label{sec:theory-evidence}
The analysis of Section~\ref{sec:detection-incentive} bears on the evaluated configuration in three ways. \emph{Capacity.} With $N=8$ and $k=1.5$, Proposition~\ref{prop:capacity} permits at most two flags per settlement, so no settlement can flag all three minority or all five majority adversaries, and any exclusion proceeds over several settlements. At $\lambda=1$ the sufficient threshold $P^*=1/3$ of Theorem~\ref{thm:divergence} lies below both capacity ceilings, $2/3$ and $2/5$, so exclusion is attainable at both ratios, with a wide margin for the minority and almost none for the majority. In the majority regime the mechanism is therefore expected to act mainly through the balance-proportional weights of~\eqref{eq:aggregation-weight}, which cap the retained adversaries' direct contribution by~\eqref{eq:honest-update-limit}, rather than through completed removal. This is consistent with the lower majority-ratio SR in Table~\ref{tab:results-adversarial}(b).

\emph{Compensation.} Compensation raises the unflagged gain $b$ of Lemma~\ref{lem:drift} from $\beta s$ to $\beta s+\lambda sF_t/(n_t-F_t)$, so balances, and with them the weights of~\eqref{eq:aggregation-weight}, shift toward unflagged agents faster as $\lambda$ grows. The evaluated adversaries follow fixed disruptive behaviors, so this reweighting channel, rather than the deterrence of Theorem~\ref{thm:incentive-compatibility}, is what the compensation study exercises; the monotone $\lambda$ ordering in Table~\ref{tab:results-compensation} and the faster loss decline in Fig.~\ref{fig:m3_policyloss} are consistent with it. Capacity also limits the channel: at the initial population, at least three of the five majority adversaries remain unflagged at every settlement and receive compensation. Averaged over the twelve cells, $\lambda=1$ reduces CO relative to $\lambda=0$ by $5.4$~pp at the $37.5\%$ training ratio and by $3.8$~pp at $62.5\%$, in the direction this leakage predicts.

\emph{Exclusion and retention.} Lemma~\ref{lem:drift} holds at every evidence-bearing settlement whatever $S_{\min}$ is, but the reported $S_{\min}=0.1s$ lies outside Assumption~\ref{as:exclusion}: a retained balance below $s$ cannot fund the next deposit. If such a balance ends participation, $s$ is the effective retention threshold and~\eqref{eq:exclusion-time} applies with $S_{\min}$ replaced by $s$, giving $\mathbb E[T^{\rm ex}_j]\le S_j(0)/\varepsilon$, at most $4S_j(0)$ settlements for $\underline p=0.5$ and $\lambda=1$; otherwise the bound remains a conditional design result for these trials. Corollary~\ref{cor:exclusion} describes the regime in which updates draw on retained honest contributors. The $100\%$ SR at $N=8$ in Table~\ref{tab:results-adversarial} and the goal-arrival plateau at five, the honest count, in Fig.~\ref{fig:m3_goal} are the outcomes that regime produces; they are consistent with it without identifying when, or whether, individual adversaries were excluded. Nonnegative honest expected drift likewise does not rule out a losing trajectory.

\section{Conclusion}
\label{sec:discussion-conclusion}
\label{sec:future-work}
BE-CTDE makes blockchain an active training-management layer for cooperative learning. A trajectory residual links executed behavior to an auditable token mechanism, and retained balances set each agent's contribution to the shared policy update, so behavioral verification, economic accountability, and learning influence operate within one CTDE workflow. Experiments on multi-agent path finding show that BE-CTDE succeeds far more often under attack than the state-of-the-art SCRIMP planner, including with an adversarial majority, and consistently collides less than variants that only flag deviations or that burn forfeited stake instead of redistributing it. Future work will extend BE-CTDE to broader threat models, further task domains, and deployment in multi-operator IoT systems.

\bibliographystyle{IEEEtran}
\bibliography{ref}


\end{document}

%% file: revision_generated/detector_capacity_table.tex
\begin{table}[t]
\centering\small
\setlength{\tabcolsep}{5pt}
\caption{Initial detector capacity at $k=1.5$. The last columns bound the realized fraction of adversaries flagged in one round at the stated initial ratios; they are not conditional detection guarantees.}
\label{tab:detector-capacity}
\begin{tabular}{@{}cccc@{}}
\toprule
$N$ & $F_{\max}(N)$ & $M/N=3/8$ & $M/N=5/8$ \\
\midrule
8 & 2 & 0.6667 & 0.4000 \\
32 & 9 & 0.7500 & 0.4500 \\
64 & 19 & 0.7917 & 0.4750 \\
128 & 39 & 0.8125 & 0.4875 \\
\bottomrule
\end{tabular}
\end{table}

%% file: revision_generated/restored_original_adversarial_table.tex
\begin{table*}[t]
\caption{Performance under adversarial training and testing. Panels show (a) $37.5\%$ and (b) $62.5\%$ malicious agents; metric columns correspond to obstacle densities $0\%$, $15\%$, and $30\%$. CO and SR refer to honest agents. Each condition reports 100 evaluation episodes. Bold highlights full BE-CTDE.}
\label{tab:results-adversarial}
\centering
\footnotesize
\renewcommand{\arraystretch}{1.05}
\setlength{\tabcolsep}{3pt}
\begin{minipage}[t]{0.49\textwidth}
\centering
\textbf{(a) $37.5\%$ malicious agents}\\[2pt]
\begin{tabular}{@{}c l c c c c c c@{}}
\toprule
& & \multicolumn{3}{c}{\textbf{CO}~$\downarrow$ (\%)} & \multicolumn{3}{c}{\textbf{SR}~$\uparrow$ (\%)} \\
\cmidrule(lr){3-5}\cmidrule(lr){6-8}
$N$ & Method & $0\%$ & $15\%$ & $30\%$ & $0\%$ & $15\%$ & $30\%$ \\
\midrule
\multirow{3}{*}{8}   & SCRIMP           & 2.47  & 3.12  & 3.96  & 0   & 0   & 0   \\
                     & BE-CTDE Lite     & 10.81 & 19.36 & 23.18 & 97  & 97  & 96  \\
                     & \textbf{BE-CTDE} & \textbf{8.77}  & \textbf{17.12} & \textbf{21.82} & \textbf{100} & \textbf{100} & \textbf{100} \\
\midrule
\multirow{3}{*}{32}  & SCRIMP           & 10.92 & 17.26 & 20.73 & 0   & 0   & 0   \\
                     & BE-CTDE Lite     & 14.21 & 21.15 & 27.88 & 97  & 94  & 92  \\
                     & \textbf{BE-CTDE} & \textbf{10.74} & \textbf{17.13} & \textbf{23.09} & \textbf{100} & \textbf{99}  & \textbf{97} \\
\midrule
\multirow{3}{*}{64}  & SCRIMP           & 12.17 & 18.78 & 24.74 & 0   & 0   & 0   \\
                     & BE-CTDE Lite     & 15.54 & 22.68 & 29.83 & 94  & 94  & 89  \\
                     & \textbf{BE-CTDE} & \textbf{11.10} & \textbf{17.60} & \textbf{23.76} & \textbf{100} & \textbf{98}  & \textbf{91} \\
\midrule
\multirow{3}{*}{128} & SCRIMP           & 27.93 & 29.45 & 36.17 & 0   & 0   & 0   \\
                     & BE-CTDE Lite     & 18.06 & 23.97 & 39.54 & 79  & 75  & 48  \\
                     & \textbf{BE-CTDE} & \textbf{16.59} & \textbf{20.22} & \textbf{37.31} & \textbf{86}  & \textbf{81}  & \textbf{52} \\
\bottomrule
\end{tabular}
\end{minipage}\hfill
\begin{minipage}[t]{0.49\textwidth}
\centering
\textbf{(b) $62.5\%$ malicious agents}\\[2pt]
\begin{tabular}{@{}c l c c c c c c@{}}
\toprule
& & \multicolumn{3}{c}{\textbf{CO}~$\downarrow$ (\%)} & \multicolumn{3}{c}{\textbf{SR}~$\uparrow$ (\%)} \\
\cmidrule(lr){3-5}\cmidrule(lr){6-8}
$N$ & Method & $0\%$ & $15\%$ & $30\%$ & $0\%$ & $15\%$ & $30\%$ \\
\midrule
\multirow{3}{*}{8}   & SCRIMP           & 50.46 & 68.55 & 71.94 & 0   & 0   & 0  \\
                     & BE-CTDE Lite     & 16.95 & 29.80 & 40.23 & 95  & 93  & 90 \\
                     & \textbf{BE-CTDE} & \textbf{12.07} & \textbf{25.59} & \textbf{36.13} & \textbf{100} & \textbf{99}  & \textbf{94} \\
\midrule
\multirow{3}{*}{32}  & SCRIMP           & 78.51 & 82.40 & 83.63 & 0   & 0   & 0  \\
                     & BE-CTDE Lite     & 17.13 & 32.91 & 39.02 & 96  & 94  & 89 \\
                     & \textbf{BE-CTDE} & \textbf{15.79} & \textbf{27.68} & \textbf{37.41} & \textbf{100} & \textbf{97}  & \textbf{93} \\
\midrule
\multirow{3}{*}{64}  & SCRIMP           & 89.46 & 92.55 & 93.42 & 0   & 0   & 0  \\
                     & BE-CTDE Lite     & 24.76 & 37.10 & 44.24 & 91  & 89  & 78 \\
                     & \textbf{BE-CTDE} & \textbf{17.03} & \textbf{24.00} & \textbf{37.61} & \textbf{97}  & \textbf{95}  & \textbf{85} \\
\midrule
\multirow{3}{*}{128} & SCRIMP           & 90.97 & 92.79 & 95.01 & 0   & 0   & 0  \\
                     & BE-CTDE Lite     & 33.69 & 41.71 & 54.11 & 79  & 75  & 34 \\
                     & \textbf{BE-CTDE} & \textbf{28.95} & \textbf{35.08} & \textbf{51.62} & \textbf{84}  & \textbf{77}  & \textbf{38} \\
\bottomrule
\end{tabular}
\end{minipage}
\end{table*}

%% file: revision_generated/restored_original_compensation_table.tex
\begin{table}[t]
\caption{Test-time collision rate (CO, \%) with adversaries during training only and honest agents during testing. Columns under each training ratio correspond to obstacle densities $0\%$, $15\%$, and $30\%$. $\lambda=0$ is the burn-only BE-CTDE variant.}
\label{tab:results-compensation}
\centering
\footnotesize
\renewcommand{\arraystretch}{1.05}
\setlength{\tabcolsep}{2.25pt}
\begin{tabular}{@{}c l c c c c c c@{}}
\toprule
& & \multicolumn{3}{c}{\textbf{37.5\% malicious agents}} & \multicolumn{3}{c}{\textbf{62.5\% malicious agents}} \\
\cmidrule(lr){3-5}\cmidrule(lr){6-8}
$N$ & Method & $0\%$ & $15\%$ & $30\%$ & $0\%$ & $15\%$ & $30\%$ \\
\midrule
\multirow{4}{*}{8}   & On-Chain CTDE                     & 2.45  & 3.09  & 3.67  & 49.67 & 58.50 & 69.72 \\
                     & BE-CTDE $\lambda{=}0$             & 9.52  & 18.14 & 29.63 & 13.51 & 29.16 & 37.08 \\
                     & BE-CTDE $\lambda{=}0.5$           & 7.81  & 16.35 & 25.82 & 13.06 & 26.22 & 37.42 \\
                     & \textbf{BE-CTDE $\lambda{=}1$}    & \textbf{7.42}  & \textbf{15.30} & \textbf{20.79} & \textbf{11.96} & \textbf{24.84} & \textbf{35.20} \\
\midrule
\multirow{4}{*}{32}  & On-Chain CTDE                     & 10.41 & 15.83 & 18.69 & 74.04 & 79.71 & 82.53 \\
                     & BE-CTDE $\lambda{=}0$             & 13.17 & 20.38 & 31.02 & 17.51 & 29.43 & 39.67 \\
                     & BE-CTDE $\lambda{=}0.5$           & 11.81 & 16.70 & 26.13 & \textbf{14.86} & 26.55 & 37.82 \\
                     & \textbf{BE-CTDE $\lambda{=}1$}    & \textbf{9.66}  & \textbf{11.59} & \textbf{20.34} & 15.19 & \textbf{25.76} & \textbf{31.12} \\
\midrule
\multirow{4}{*}{64}  & On-Chain CTDE                     & 11.89 & 16.54 & 23.71 & 80.95 & 83.61 & 87.49 \\
                     & BE-CTDE $\lambda{=}0$             & 16.25 & 19.94 & 26.81 & 20.91 & 23.28 & 37.42 \\
                     & BE-CTDE $\lambda{=}0.5$           & 14.07 & 19.35 & 24.08 & 17.39 & 22.72 & 34.58 \\
                     & \textbf{BE-CTDE $\lambda{=}1$}    & \textbf{10.91} & \textbf{15.27} & \textbf{21.52} & \textbf{15.43} & \textbf{21.32} & \textbf{32.97} \\
\midrule
\multirow{4}{*}{128} & On-Chain CTDE                     & 25.41 & 29.59 & 32.73 & 81.48 & 86.16 & 89.53 \\
                     & BE-CTDE $\lambda{=}0$             & 17.35 & 26.94 & 38.16 & 29.18 & 32.28 & 43.56 \\
                     & BE-CTDE $\lambda{=}0.5$           & 14.52 & 25.74 & 38.02 & 26.91 & 29.03 & 43.15 \\
                     & \textbf{BE-CTDE $\lambda{=}1$}    & \textbf{14.28} & \textbf{20.04} & \textbf{34.92} & \textbf{25.66} & \textbf{27.49} & \textbf{40.15} \\
\bottomrule
\end{tabular}
\end{table}